\documentclass[onecolumn,11pt]{article}
\usepackage[top=2.54cm, bottom=2.54cm, left=2.54cm, right=2.54cm,a4paper]{geometry}

\usepackage[utf8]{inputenc} % Any characters can be typed directly from the keyboard, eg éçñ

\usepackage{amsmath,amsthm}
\usepackage{bm}  % Define \bm{} to use bold math fonts
\usepackage{bbm}

\usepackage{tikz,pgfplots}
\pgfplotsset{compat=1.14}

\usepackage[backend=bibtex,sorting=none,style=phys,biblabel=brackets,backref=false,bibencoding=utf8,maxnames=20,eprint=true]{biblatex}
\makeatletter
\pretocmd{\blx@head@bibintoc}{\phantomsection}{}{\ddt}
\makeatother
\DefineBibliographyStrings{english}{%
  backrefpage = {page},% originally "cited on page"
  backrefpages = {pages},% originally "cited on pages"
  phdthesis = {PhD dissertation}
}
\usepackage{caption}
\usepackage{titlesec}

\titleformat*{\section}{\large \bfseries}
\titleformat*{\subsection}{\normalsize\bfseries}
\titleformat*{\subsubsection}{\bfseries}
\titleformat*{\paragraph}{\large\bfseries}
\titleformat*{\subparagraph}{\large\bfseries}

\titlespacing*\section{0pt}{12pt plus 4pt minus 2pt}{2pt plus 2pt minus 2pt}

\usepackage{xcolor}
\definecolor{dullmagenta}{rgb}{0.4,0,0.4}   % #660066
\definecolor{darkblue}{rgb}{0,0,0.4}
\usepackage[bookmarks,colorlinks,breaklinks]{hyperref}  % PDF hyperlinks, with coloured links
\hypersetup{linkcolor=red,citecolor=blue,filecolor=dullmagenta,urlcolor=darkblue} % coloured links

\usepackage[bitstream-charter,greekuppercase=italicized,cal=cmcal]{mathdesign}
\usepackage[T1]{fontenc}

\usepackage{amsbsy,verbatim,graphicx}

\usepackage{braket}
\newcommand{\ketbra}[1]{|#1\rangle\langle #1|}

\renewcommand{\epsilon}{\varepsilon}
\renewcommand{\phi}{\varphi}
\def\tr{{\rm Tr}}

\def\ot{\otimes}
\def\eps{\epsilon}

\DeclareFontFamily{U}{bbold}{}
\DeclareFontShape{U}{bbold}{m}{n}
 {
  <-5.5> s*[1.069] bbold5
  <5.5-6.5> s*[1.069] bbold6
  <6.5-7.5> s*[1.069] bbold7
  <7.5-8.5> s*[1.069] bbold8
  <8.5-9.5> s*[1.069] bbold9
  <9.5-11> s*[1.069] bbold10
  <11-15> s*[1.069] bbold12
  <15-> s*[1.069] bbold17
 }{}

\DeclareRobustCommand{\Eins}{%
  \text{\usefont{U}{bbold}{m}{n}1}%
}
\def\id{\Eins}

\usepackage{nicefrac}

\newcommand{\cme}[4]{H_{\min}^{#1}(#2\lvert#3)_{#4}}
\newcommand{\lme}[5]{\lambda_{#1}^{#2}(#3\lvert#4)_{#5}}
\usepackage{xspace}

\def\pinch{\mathcal P}
\def\cC{\mathcal C}
\def\cD{\mathcal D}
\def\cE{\mathcal E}
\def\cF{\mathcal F}

\definecolor{mblue}{rgb}{0.368417, 0.506779, 0.709798}
\definecolor{morange}{rgb}{0.880722, 0.611041, 0.142051}
\definecolor{mgreen}{rgb}{0.560181, 0.691569, 0.194885}
\definecolor{mred}{rgb}{0.922526, 0.385626, 0.209179}
\definecolor{mpurple}{rgb}{0.528488, 0.470624, 0.701351}

\newtheorem{theorem}{Theorem}
\newtheorem{prop}{Proposition}
\newtheorem{lemma}{Lemma}

\newtheorem{remark}[theorem]{Remark}
\theoremstyle{definition}
\newtheorem{definition}[theorem]{Definition}

\usepackage{titling}
\usepackage{authblk}

 \makeatletter
 \def\maketitle{\AB@maketitle}
 \makeatother

\pretitle{\begin{center}\large \bfseries}
\posttitle{\par\end{center}}
\newcommand{\cl}[2]{M^\star(#1,#2)} %symbol for compression length
\usepackage{accents}
\newcommand{\ubar}[1]{\underaccent{\bar}{#1}}
\usepackage{IEEEtrantools}
\usepackage{mathtools}

\begin{document}

\author[1]{ Dina Abdelhadi\thanks{\href{mailto:dinaa@student.ethz.ch}{dinaa@student.ethz.ch}}}
\author[2]{Joseph M.\ Renes\thanks{\href{mailto:renes@phys.ethz.ch}{renes@phys.ethz.ch}}}

\affil[1]{Department of Information Technology and Electrical Engineering, ETH Z\"urich}
\affil[2]{Institute for Theoretical Physics, ETH Z\"urich\authorcr {\rmfamily\small \today}}%\authorcr {\rmfamily\slshape\small ETH Z\"urich, Z\"urich, Switzerland}}

% \author{{\normalsize Dina Abdelhadi and Joseph M.~Renes}\\
% { \normalsize \emph{Institute for Theoretical Physics, ETH Z\"urich, 8093 Z\"urich, Switzerland}}\\
% {\normalsize \today}
% }

\date{}

\title{On the Second-Order Asymptotics of the Partially Smoothed Conditional Min-Entropy \& Application to Quantum Compression}

%\date{\today}
%\date{\today}
%\end{center}
%\vspace{-\baselineskip}}

\maketitle

\vspace{-3\baselineskip}
\begin{abstract}
Recently, Anshu \emph{et al.}\  introduced ``partially'' smoothed information measures and used them to derive tighter bounds for several information-processing tasks, including quantum state merging and privacy amplification against quantum adversaries [arXiv:1807.05630 [quant-ph]]. 
Yet, a tight second-order asymptotic expansion of the partially smoothed conditional min-entropy in the i.i.d.\ setting remains an open question. 
Here we establish the second-order term in the expansion for pure states, and find that it differs from that of the original ``globally'' smoothed conditional min-entropy. 
Remarkably, this reveals that the second-order term is not uniform across states, since for other classes of states the second-order term for partially and globally smoothed quantities coincides. 
By relating the task of quantum compression to that of quantum state merging, our derived expansion allows us to determine the second-order asymptotic expansion of the optimal rate of quantum data compression. 
This closes a gap in the bounds determined by Datta and Leditzky [IEEE Trans.\ Inf.\ Theory 61, 582 (2015)], and shows that the straightforward compression protocol of cutting off the eigenspace of least weight is indeed asymptotically optimal at second order.% \cite{datta2015second}.

\end{abstract}

\vspace{0mm}

 %!TEX root = paper.tex
\section{Introduction}
Finding tight bounds on the performance and costs of quantum information processing tasks is a central research area in quantum information theory. Examples include bounds on the number of entangled pairs that must be shared between two parties to achieve quantum state merging, or the maximum secret key-length that can be extracted in the presence of an adversary with quantum knowledge about the initial string, as in privacy amplification.
As in classical information theory, these bounds usually involve quantifying uncertainty and information through the various of notions of entropy. 
A prominent example relevant in the two aforementioned protocols is the quantum version of the conditional min-entropy of the bipartite quantum density operator $\rho_{AR}$, 
%the smallest of the family of the R\'{e}nyi entropies, which is defined as follows
\begin{align}
H_{\min}(A\lvert R)_\rho \coloneqq -\inf \{ \lambda:\rho_{AR}\leq 2^\lambda \id_A\otimes\rho_R\}\,.
\end{align}

Smoothing this quantity by an amount $\eps$ replaces $\rho_{AR}$ (but not $\rho_R$) in the optimization by a $\sigma_{AR}$ such that $\Delta(\rho_{AR},\sigma_{AR})\leq \eps$ for some distance measure $\Delta$. 
The smoothed conditional min-entropy was introduced by Renner~\cite{Renner2016}, and shown to provide bounds on the optimal key length of privacy amplification \cite{Tomamichel2013} and optimal entanglement cost of state merging \cite{Berta2009}. 
Recently, Anshu \emph{et al.}\ observed that further constraints on the smoothing yield improved bounds for these two protocols~\cite{Anshu2018}. 
In particular, they defined the \emph{partially-smoothed} conditional min-entropy to consider only those $\sigma_{AR}$ that are $\eps$-close to $\rho_{AR}$ and satisfy the marginal constraint $\sigma_R\leq \rho_R$. 
The resulting quantity can be expressed as the optimal value of a semidefinite program (SDP). 

Smoothing with additional constraints has also been considered in the classical domain. 
Renner and Wolf imposed the additional constraint $R_{XY}\leq P_{XY}$ when smoothing the Renyi entropy over nearby distributions $R_{XY}$~\cite{renner_smooth_2004,renner_simple_2005}. 
In \cite[Appendix I]{Yang2017}, Yang, Schaefer, and Poor implicitly consider the globally-smoothed version of the R\'enyi entropy of order 2 and showed that it is equivalent to the version of Renner and Wolf.

The asymptotics of the globally smoothed conditional min-entropy is well-understood for i.i.d.\ states. 
In particular, for $\Delta$ the purification distance $P$, Tomamichel and Hayashi show in \cite{Tomamichel2013} that 
\begin{align}
 H_{\min}^{\eps,P}(A\lvert R)_{\rho^{\otimes n}} = nH(A|R)_\rho+\sqrt {n V(A|R)_\rho}\Phi^{-1}(1-\eps^2)+o(n)\,,
\end{align}
where $H(A|R)_\rho$ is the conditional von Neumann entropy, $V(A|R)_\rho$ is the conditional variance, and $\Phi^{-1}$ is the quantile function of the $\mathcal N(0,1)$ normal distribution (for more precise definitions, see \S\ref{sec:def}). 
For the partially smoothed entropy, denoted $H_{\min}^{\eps,P}(A|\dot R)_\rho$ (note the dot over $R$), \cite{Anshu2018} established that the coefficient of the first term (the first order term) is again the von Neumann entropy, but the second order term is not known. 
The methods needed to pin down the form of the second order term have application beyond the asymptotic limit, as they also often lead to tight bounds for finite $n$, as vividly demonstrated by Polyanskiy \emph{et al.}~\cite{polyanskiyChannelCodingRate2010}.

\paragraph{\normalsize Outline \& Contributions}
%In this paper we study the asymptotic expansion of the partially smoothed conditional min-entropy and determine the second-order asymptotics of quantum data compression. 
The main contributions of this paper can be summarized as follows: 
\begin{itemize}
\item In Section \ref{sec:pscme} we show that the SDP defining the partially smoothed conditional min-entropy for pure states can be reduced to a quadratic program, and indeed to a convex optimization in a single real variable.

\item In Section \ref{subsec:asym}, we prove the central result of this paper, which is the following second-order asymptotic expansion of the partially smoothed conditional min-entropy for pure $\rho_{AR}$,	
\begin{align}
	-H_{\min}^{\epsilon,P}(A\lvert\dot R)_{\rho_{AR}^{\otimes n}} = nH(A)_\rho+\sqrt{nV(A)_\rho}\Phi^{-1}(\sqrt{1-\epsilon^2}) +{O}(\log n).
\end{align}
We obtain this expansion by bounding the quadratic program from above and below by the hypothesis testing quantity $\beta_{\sqrt{1-\epsilon^2}}(\rho_A, \mathbbm 1)$ and using the second-order expansion for the latter from \cite[Eq. 34]{Tomamichel2013} or \cite{li_second-order_2014}. 
Thus, in this case the second order term of the partially smoothed entropy differs from its globally smoothed counterpart.
We further show that there exist other classes of states for which the asymptotic behavior of partial and global smoothing coincides, implying the second order term is not uniform across all states.

% As a consequence,  the partially smoothed conditional min-entropy in the pure state case inherits the expansion \cite[equation 34]{Tomamichel2013}  known for the hypothesis testing relative entropy $\beta_{\sqrt{1-\epsilon^2}}(\rho_A, \mathbbm 1)$. It is due to this relation that the form $\sqrt{1-\epsilon^2}$ shows up in the second order term, instead of ${1-\epsilon^2}$, as in the expansion of the globally smoothed conditional min-entropy.

%\item Next, we derive second-order asymptotic expansions for a few simple cases aside from the pure state case. These expansions turn out to have a second order similar to the globally smoothed quantity rather than $\sqrt{1-\epsilon^2}$ like the pure state case. Remarkably, this implies that the general asymptotic expansion of the partially smoothed conditional min-entropy is not uniform across states. 
\item In Section \ref{sec:qc}, we use the derived pure state asymptotic expansion to find bounds on quantum data compression that are tight at second-order. 
The bounds reported in \cite[Theorem 5.8]{datta2015second} do not match at second order, and we find that the loose bound is the converse. 
The improved converse bound is obtained by observing that any compression protocol can be used to construct a state merging protocol, and therefore compression inherits the state merging converse in terms of partially-smoothed conditional min-entropy from  \cite{Anshu2018}. 
This implies that it is not necessary to make use of quantum coherence in any way to obtain an asymptotically optimal second-order compression rate, and the straightforward protocol of cutting off the eigensubspace of lowest weight is optimal.  
%Towards this goal, we explain a connection between the task of quantum compression and that of state merging, which implies a relation between the bounds of the two tasks. Since the bounds recently derived for quantum state merging \cite{Anshu2018} are in terms of the partially smoothed conditional min-entropy, the expansion derived in \S\ref{subsec:asym} proves useful in closing the gap in the second-order term of the quantum compression bounds derived in \cite[Theorem 5.8]{datta2015second}.

\end{itemize}

\section{Mathematical setup}\label{sec:def}
In this section we define several quantities and set the notation that will be used throughout this paper. 
We denote the positive semidefinite operators on the vector space $\mathbb C^d$ by $\mathbb P_d$. 
The inequality $\rho\geq \sigma$ between two linear operators $\rho$ and $\sigma$ denotes $\rho-\sigma\in \mathbb P_d$. 
By $\{\rho-\sigma\geq 0\}$, we denote the projector onto the positive part of $\rho-\sigma$. 
The identity operator is denoted by $\id$ and we denote the elements of a  ``standard'' basis for $\mathbb C^d$ by $\ket{x}$, where $x$ ranges from $0$ to $d-1$. Arithmetic inside the ket is taken modulo $d$. 
The particular dimension $d$ plays no role here, except that it is finite. 
Quantum states are those $\rho\in \mathbb P_d$ satisfying $\tr[\rho]=1$, referred to as ``normalized''. Subnormalized states $\sigma\in \mathbb P_d$ satisfying $\tr[\sigma]\leq 1$ will also be of use. 
We follow the convention of labelling subsystems, i.e.\ tensor factors of states defined on tensor product spaces $\mathbb C^{d_A\times d_B}$, and denote the linear operator $\rho$ on subsystems $A$ and $B$ by $\rho_{AB}$. In this convention, which mimics the labelling of probability distributions by random variables, the partial trace of $\rho_{AB}$, say over $A$, is simply denoted by $\rho_{B}$.   
We make extensive use of semidefinite programming methods; for an overview see \cite{watrous2018theory}.
Two distance metrics will be of particular importance here, the purification distance and trace distance. 
To define them, first set $\tr[|A|]=\tr[\sqrt{A^* A}]$ for any linear operator $A$ on $\mathbb C^d$ and $A^*$ its adjoint. 
%First, we review the definitions of the metrics $\Delta\in\{P,T\}$ that are used throughout the paper to quantify the distance between two states.
\begin{definition}[Purified Distance \cite{Tomamichel2010}] 
For $\rho,\sigma\in \mathbb P_d$, the purified distance $P$ is defined in terms of the generalized fidelity $F$:
	\[P(\rho,\sigma)\coloneqq \sqrt{1-F^2(\rho,\sigma)}\text{, where }F(\rho,\sigma)\coloneqq \tr\left[\lvert \sqrt{\rho} \sqrt{\sigma}\rvert\right]+\sqrt{(1-\tr[ \rho])(1-\tr[\sigma])} \,.\]
\end{definition}
\begin{definition}[Trace Distance] For $\rho,\sigma\in\mathbb P_d$, the generalized trace distance is given by 
	\[T(\rho,\sigma) \coloneqq \tfrac{1}{2} \tr\left[\lvert \rho -\sigma\lvert\right]+\tfrac{1}{2} \left\lvert\tr[\rho-\sigma]\right\lvert.\]
\end{definition}
Note the following equivalent reformulation of the definition of the trace distance.
\begin{remark}[SDP formulation of trace distance]\label{SDPform}
For normalized $\rho$ and subnormalized $\sigma$ both in $\mathbb P_d$, 
%Consider a normalized $\rho$ and a sub-normalized $\sigma$, then the generalized trace distance is equivalent to 
\[T(\rho,\sigma) = \max_{0\leq\Lambda\leq\mathbbm{1}}\tr[\Lambda(\rho-\sigma)]=\min_{\tau \geq\rho-\sigma, \tau\geq 0} \tr[\tau],\] where the last equality follows by duality.
\end{remark}

%%% ?? Do I need to prove this? 
%For the rest of this paper, we take $0\leq \epsilon\leq 1$.

%Now, we restate the definitions of the globally and partially smoothed conditional min-entropy, as introduced in \cite{Anshu2018}.
%First let $\mathcal B_{\eps,\Delta}(\rho_{AR}):=\{\sigma_{AR}:\Delta(\rho_{AR},\sigma_{AR})\leq \eps,\sigma_{AR}\geq 0, \tr[\sigma_{AR}]\leq 1\}$ and $\mathcal B^{\text{part}}_{\eps,\Delta}(\rho_{AR}):=\{\sigma_{AR}:\sigma_{AR}\in \mathcal B_{\eps,\Delta}(\rho_{AR}),\sigma_R\leq \rho_R\}$.
\begin{definition}[Smoothed Conditional Min-entropy] 
For a distance measure $\Delta$ and arbitrary normalized state $\rho_{AR}\in \mathbb P_{d_Ad_R}$,  the partially smoothed conditional min-entropy is defined by 
\[\cme{\epsilon,\Delta}{A}{\dot{R}}{\rho_{AR}} \coloneqq -\log \lme{\max}{\epsilon,\Delta}{A}{\dot{R}}{\rho_{AR}}\,,\]
 while the globally smoothed conditional min-entropy is defined by 
 \[\cme{\epsilon,\Delta}{A}{{R}}{\rho_{AR}} \coloneqq -\log \lme{\max}{\epsilon,\Delta}{A}{{R}}{\rho_{AR}}\,,\] where
 \begin{align}
 \label{pse0}
 \lambda_{\max}^{\epsilon,\Delta}(A\lvert \dot{R})_{\rho_{AR}}\coloneqq \begin{array}[t]{rl}\underset{\lambda,\sigma}{\text{infimum}} & \lambda\\
 \text{subject to} & \lambda \id_A\otimes \rho_R\geq \sigma_{AR}\\
 & \Delta(\rho_{AR},\sigma_{AR})\leq \eps\\
 & \sigma_R\leq \rho_R\\
 &\lambda\in \mathbb R_+,\sigma_{AR}\in \mathbb P_{d_Ad_R}\,,
 \end{array}
 %|\lambda\id_A\otimes\rho_R\geq \sigma_{AR},,\sigma_R\leq \rho_R,\lambda\in \mathbb R_+,\sigma_{AR}\in \mathbb P(d_Ad_R)\}
 \end{align}
 and
 \begin{align}
 \label{gse}
 \lambda_{\max}^{\epsilon,\Delta}(A\lvert {R})_{\rho_{AR}}\coloneqq \begin{array}[t]{rl}\underset{\lambda,\sigma}{\text{infimum}} & \lambda\\
 \text{subject to} & \lambda \id_A\otimes \rho_R\geq \sigma_{AR}\\
 & \Delta(\rho_{AR},\sigma_{AR})\leq \eps\\
 & \tr[\sigma_R]\leq 1\\
 &\lambda\in \mathbb R_+,\sigma_{AR}\in \mathbb P_{d_Ad_R}\,.
 \end{array}
 %|\lambda\id_A\otimes\rho_R\geq \sigma_{AR},,\sigma_R\leq \rho_R,\lambda\in \mathbb R_+,\sigma_{AR}\in \mathbb P(d_Ad_R)\}
 \end{align}

% \begin{minipage}[]{.4\textwidth}
% 	\begin{equation} {\label{pse0}} 
% \begin{split}       	       
% 	\lambda_{\max}^{\epsilon,\Delta}(A\lvert \dot{R})_{\rho_{AR}}=&	\inf_{\lambda,\sigma_{AR}\geq 0}{\lambda} \\  % objective function and label
% 	&\text{subject to}
% 	\\&{\lambda \mathbbm{1}_A \otimes \rho_R}{\geq \sigma_{AR}}    % constraint 1
% 	\\&{\pmb{\sigma_R}}{\pmb{\leq \rho_R} }  % constraint 2
% 	\\&{\Delta(\sigma_{AR},\rho_{AR})}{\leq \epsilon,  }  % constraint 2
% \end{split}
% \end{equation}
% \end{minipage}
% \begin{minipage}[]{.4\textwidth}
% 	\begin{equation} {\label{gse}} 
% \begin{split}       	       
% \lambda_{\max}^{\epsilon,\Delta}(A\lvert {R})_{\rho_{AR}}=&	\inf_{\lambda,\sigma_{AR}\geq 0}{\lambda} \\  % objective function and label
% &\text{subject to}
% \\&{\lambda \mathbbm{1}_A \otimes \rho_R}{\geq \sigma_{AR}}    % constraint 1
% \\&{\pmb{\tr[\sigma_{AR}]}}{\pmb{\leq 1} }  % constraint 2
% \\&{\Delta(\sigma_{AR},\rho_{AR})}{\leq \epsilon,  }  % constraint 2
% \end{split}
% \end{equation}
% \end{minipage}
Note that the difference between the two semidefinite programs is the relaxation of the final constraint (setting aside the domain specification of the variables) $\sigma_R \leq \rho_R$ to $\tr[\sigma_{AR}]\leq 1$.
When $\Delta$ is the purification or trace distance, the above optimizations can be expressed as SDPs. 

\end{definition}
The following quantities will also be used.
\begin{definition}[Information spectrum entropies]
	For a normalized $\rho$ and any $\sigma$, both in $\mathbb P(d)$, the information spectrum relative entropies of \cite{datta2015second} are defined 
	as follows:
	\[\ubar D_s^\epsilon(\rho \lVert \sigma) \coloneqq \sup\,\{\gamma \mid \tr[(\rho-2^\gamma \sigma) \{\rho > 2^\gamma \sigma\}]\geq 1-\epsilon\}\,,\] 
	\[\bar D_s^\epsilon(\rho \lVert \sigma) \coloneqq \inf\{\gamma \mid \tr[(\rho-2^\gamma \sigma) \{\rho > 2^\gamma \sigma\}]\leq \epsilon\}\,.\]
	The information spectrum entropies are consequently defined as 
	\[\ubar{H}_s^\epsilon(\omega) \coloneqq -\bar D_s^\epsilon(\omega \lVert \mathbbm 1), \quad \bar{H}_s^\epsilon(\omega) \coloneqq -\ubar D_s^\epsilon(\omega \lVert \mathbbm 1).\]
	
\end{definition}

We note in passing that $E_\gamma$-divergence of Liu \emph{et al.}~\cite{Liu2015} is equivalent to this version of the information spectrum entropy. 
It is defined as $E_\gamma(\rho,\sigma) \coloneqq \tr[(\rho-\gamma \sigma)\{\rho > \gamma\sigma\}]$, and therefore 
\[\bar{D}_s^\epsilon(\rho \lVert \sigma) = \inf\{\gamma \lvert E_{2^\gamma}(\rho,\sigma) \leq \epsilon \} \quad\text{and}\quad \ubar{D}_s^\epsilon(\rho \lVert \sigma) = \sup\{\gamma \lvert E_{2^\gamma}(\rho,\sigma) \geq 1-\epsilon \}\,.\]

\begin{definition}[Hypothesis testing quantity]\label{def:beta}
For normalized $\rho$ and arbitrary $\sigma$, both in $\mathbb P_d$, 
	\[\beta_\alpha(\rho,\sigma) \coloneqq \min_{0\leq \Lambda \leq \mathbbm 1}\{\tr [\Lambda\sigma] \mid \tr[\Lambda\rho]\geq \alpha\}.\] %Strong Duality holds, see QIT notes
\end{definition}
\begin{definition}[Relative entropy and Variance]
For normalized $\rho$ and arbitrary $\sigma$, both in $\mathbb P_d$, 
the relative entropy is $D(\rho,\sigma)\coloneqq \tr[\rho(\log \rho-\log \sigma)]$, and the variance is $V(\rho,\sigma)\coloneqq \tr[\rho(\log \rho-\log\sigma)^2]-D(\rho,\sigma)^2$. 
\end{definition}
\begin{definition}
	[Quantile of the standard normal distribution $\Phi^{-1}(x)$]\[\Phi^{-1}(x)\coloneqq \sup\left\{z \in \mathbb{R}\left \lvert \frac{1}{\sqrt{2\pi}}\int_{-\infty}^{z}e^{-t^2/2}dt\leq x\right.\right\}.\]
\end{definition}

\begin{definition}[Conditional Entropy and Variance]
For a normalized $\rho_{AR}\in \mathbb P_{d_Ad_R}$, the conditional entropy is $H(A|R)_\rho \coloneqq -D(\rho_{AR},\id_A\otimes \rho_R)$, while the conditional variance is (abusing notation slightly) $V(A|R)_\rho\coloneqq V(\rho_{AR},\id_A\otimes \rho_R)$. 
\end{definition}
%Note that for a pure state $\rho_{AR}$ we have $H(A|R)_\rho=H(A)_\rho$ and $V(A|R)_\rho=V(A)_\rho$.

In \cite{Tomamichel2013,li_second-order_2014}, it is shown that for all $\alpha\in (0,1)$ and quantum states $\rho$ and $\sigma$,
\begin{align}
\label{eq:beta2nd}
-\log \beta_\alpha(\rho^{\otimes n},\sigma^{\otimes n})=n D(\rho,\sigma)-\sqrt{n V(\rho,\sigma)}\,\Phi^{-1}(\alpha)+O(\log n)\,.
\end{align}
From this expansion, \cite[Proposition 4.9]{datta2015second} derives the following for  $\eps\in (0,1)$ and all states $\rho$ and $\sigma$:
\begin{align}
\ubar D_s^\epsilon(\rho^{\ot n} \lVert \sigma^{\ot n}) & =n D(\rho,\sigma)+\sqrt{n V(\rho,\sigma)}\,\Phi^{-1}(\eps)+O(\log n)\,,\label{eq:spec2nd}\\
\bar D_s^\epsilon(\rho^{\ot n} \lVert \sigma^{\ot n}) & =n D(\rho,\sigma)-\sqrt{n V(\rho,\sigma)}\,\Phi^{-1}(\eps)+O(\log n)\,.
\end{align}

%\noindent\hrulefill 

 %!TEX root = paper.tex

% \section{Partially Smoothed Conditional Min-Entropy} \label{sec:pscme}
% The SDP defining the partially smoothed conditional min-entropy is computationally expensive to solve, requiring long times and substantial computing power even for states of dimensions above a few tens. Thus, reducing the problem is a significant step in studying this quantity. We will focus, mainly, on the pure state case when the distance metric is the purified distance, since this is the relevant case to the setting of quantum compression, discussed later in this paper. We also look at other special cases, to gain a sense of what the general asymptotic expansion might look like and to get some understanding of the effect of introducing the partial smoothing constraint.

\section{Reduction for pure states}
\label{sec:pscme}

%\subsection{Reduction of the SDP to a Quadratic Program}
In this section we show a reduction of \eqref{pse0} to a quadratic program, for the case of pure $\rho_{AR}$ and purification distance smoothing. 
The reduction will be demonstrated in two steps. First, we show a reduction of the bipartite SDP (\ref{pse0}) to an SDP involving only a single system.
Then it is easier to show that the single system SDP is equivalent to a quadratic program.

For the first step, first note that when $\rho_{AR}$ is pure, the purification distance constraint $P(\rho_{AR},\sigma_{AR})\leq \eps$ is equivalent to $\tr[\rho_{AR}\sigma_{AR}]\geq 1-\eps^2$. 
%Furthermore, for $\rho_{AR}$ pure, we can take the dimension $d_R$ of $R$ to be 
Then the optimization in \eqref{pse0} is the SDP
\begin{align}
\label{pse}
%\lambda_{\max}^{\eps,P}(A|R)_\rho=
f_{\max}(\rho_{AR},\eps)\coloneqq\inf\{\lambda:\lambda\id_A\otimes \rho_R\geq \sigma_{AR},\tr[\rho_{AR}\sigma_{AR}]\geq 1-\eps^2,\sigma_R\leq \rho_R,\lambda\in\mathbb R_+,\sigma_{AR}\in \mathbb P_{d_A d_R}\}\,.
\end{align}
We can just as well use the Schmidt decomposition and take the standard basis of $A$ and $R$ to be the respective Schmidt bases, so that $\ket{\rho}_{AR}=\sum_x\sqrt{p_x}\ket{xx}_{AR}$ for some set of Schmidt coefficients $p_x$. 
These form a probability distribution, call it $P_X$, and since the optimization is specified by $P_X$, we write $f_{\max}(P_X,\eps)$. We shall also write $p_x$ for $P_X(x)$.  

For any given probability distribution $P_X$, let 
$\ket{\psi_X}=\sum_x {p_x}\ket{x}$
%$\psi(X)=\ketbra{\psi(X)}$ for $\ket{\psi(X)}$ 
be the unnormalized superposition in the standard basis. 
Abusing notation somewhat, denote $\ketbra{\psi_X}$ by $\psi_X$ (so that $X$ is part of the name of the operator, and not a system label). 
%, $\ket{\psi(X)}=\sum_x {P_X(x)}\ket{x}$.
Furthermore, let $\pinch$ be the pinching or diagonalization map given by $\pinch:\sigma \mapsto \sum_x \ketbra{x} \bra{x}\sigma\ket{x}$, with which we can define the single-system SDP
% Define the non-normalized superposition state on a single system using the squares of the Schmidt coefficients, $\ket{\psi}=\sum_k {p_k}\ket{k}$, and let $\mathsf{P}$ be the pinching map given by $\mathsf{P}:\sigma \mapsto  \sum_k \ketbra{k}\sigma\ketbra{k}$. 
% Then we can define the single-system SDP
\begin{align}
\label{spse}
f_{\text{SDP}}(P_X,\eps)\coloneqq
%\lambda_{\text{ss}}^{\epsilon}(X) \coloneqq 
\inf\{ \lambda \mid\lambda \id - \theta\geq 0, 
\pinch(\theta)\leq \id,
\tr[{\psi_X}\theta]\geq  1-\epsilon^2,\lambda\in \mathbb R_+,\theta\in \mathbb P_d
\}\,.
\end{align}
\begin{lemma}
For any distribution $P_X$, 
%For all pure $\rho_{AR}$ with distribution of Schmidt coefficients $P_X(x)$,  
\begin{align}
f_{\max}(P_X,\eps)= f_{\textup{SDP}}(P_X,\eps)\,.%\lambda_{\text{ss}}^{\epsilon}(X)\,.
\end{align}
\end{lemma}

% The following lemma proves the first reduction.

% For a pure state $\rho_{AR}$ with the Schmidt decomposition  $\ket{\rho_{AR}} = \sum_k \sqrt{p_k} \ket{kk}$, the SDP (\ref{pse}) with the purified distance metric is equivalent to the single system SDP (\ref{spse}), i.e., \[\lambda_{\max}^{\epsilon,P}(A\lvert \dot{R})_{\rho_{AR}}= \lambda_{ss}^{\epsilon,P}({\psi}_R),\text{ where}\] $\ket{\psi}_R \coloneqq \sum_k \sqrt{p_k} \ket{k}$, $\mathsf{P}$ is the pinching map given by $\mathsf{P}(.) \coloneqq \sum_k \ketbra{k}(.)\ketbra{k}$, $\psi^\prime =\mathsf{P}(\psi)^{1/2}\psi_R\mathsf{P}(\psi)^{1/2} $ and

% \begin{equation}{\label{spse}}  
% \lambda_{ss}^{\epsilon,P}({\psi}_R) \coloneqq \min_{\lambda\geq 0,s\geq 0}\{ \lambda \mid\lambda \mathbbm{1} - s\geq 0, 
% \mathsf{P}(s)\leq \mathbbm{1},
% \tr[{\psi}^\prime s]\geq  1-\epsilon^2
% \}.
% \end{equation}

% \end{lemma}

\begin{proof}
The first step is to establish that if the pair $(\sigma_{AR},\lambda)$ is feasible for \eqref{pse}, then  $(\Pi\sigma_{AR}\Pi,\lambda)$ is as well, where $\Pi \coloneqq \sum_x \ketbra{xx}$ in the Schmidt basis. 
To see this, consider the various constraints in turn. 
Conjugating the first by $\Pi$ produces $\lambda\id_A\otimes \rho_R\geq \Pi\sigma_{AR}\Pi$, since $\rho_R$ is diagonal in the $\ket{k}$ basis. 
The second constraint holds because $\Pi\rho_{AR}\Pi=\rho_{AR}$. 
For the third, note that $\tr_A[\Pi\sigma_{AR}\Pi]=\sum_x \ketbra{x}_R \bra{xx}\sigma_{AR}\ket{xx}$, while $\pinch(\sigma_R)=\sum_{xy}\ketbra{x}_R \bra{yx}\sigma_{AR}\ket{yx}$. 
Hence $\tr_A[\Pi\sigma_{AR}\Pi]\leq \pinch(\sigma_R)$. 
Applying $\pinch$ to the third constraint gives $\pinch(\sigma_R)\leq \rho_R$, and therefore $\tr_A[\Pi\sigma_{AR}\Pi]\leq \rho_R$ as intended. 
Finally, the last two constraints are clearly satisfied. 

Now we may restrict attention to feasible pairs $(\sigma_{AR},\lambda)$ such that $\Pi\sigma_{AR}\Pi=\sigma_{AR}$, i.e.\ supported only on the ``diagonal'' subspace $\text{span}(\{\ket{xx}\})$. 
These states are equivalent to states on a single system: For $U \coloneqq  \sum \ket{y-x}\bra{y}_A \otimes \ketbra{x}_R$, any such $\sigma_{AR}$ defines a $\theta_R$ via $\theta_{R}=\tr_A[\ketbra{0}_A U\sigma_{AR} U^*]$, while any $\theta_R$ defines a $\sigma_{AR}$ via $U^*(\ketbra{0}_A\otimes \theta_R)  U$.  
Using  $U$ we can simplify \eqref{pse}. 
Conjugating the first constraint by $U$ produces $\lambda\id_A\otimes \rho_R\geq \ketbra{0}_A\otimes \theta_R$, which implies $\lambda\rho_R\geq \theta_R$. 
The second constraint is equivalent to $\tr[\varphi \theta]\geq 1-\eps^2$, where $\ket{\varphi}=\sum_x \sqrt{p_x}\ket{x}$. 
Finally, in the third we have $\pinch(\theta_R)\leq \rho_R$. 
Therefore we have  
\begin{align}
f_{\max}(P_X,\eps)=\inf\{\lambda:\lambda\rho_R\geq\theta_R,\tr[\varphi_R\theta_R]\geq 1-\eps^2,\pinch(\theta_R)\leq \rho_R,\lambda\geq 0,\theta_{A}\geq 0\}\,.
\end{align}
This optimization is equivalent to \eqref{spse} upon substituting $\theta\to \rho_R^{1/2}\theta\rho_R^{1/2}$ and using $\pinch(\rho_R)=\rho_R$. 
\end{proof}

Now consider the following quadratic program, defined for any probability distribution $P_X$, 
\begin{equation} \label{qp}
f_{\text{QP}}(P_X,\eps) \coloneqq \inf_{g_k}\left\{\sum_x g_x^2 \left\lvert \sum_x g_{x} p_x \geq \sqrt{1-\epsilon^2}\right.,0\leq g_x\leq 1
\right\}.
\end{equation}
Before launching into the equivalence, let us first consider the optimizers of $f_{\text{QP}}(P_X,\eps)$. 
The Karush–Kuhn–Tucker (KKT) conditions (see \cite[\S5.5.3]{boyd2004convex}) are simply 
\[
	-2 g_x=-a p_x+b_x-c_x;a\left(\sqrt{1-\epsilon^2}-\sum g_x p_x\right)= 0;b_x(g_x-1)= 0;-c_xg_x= 0 ; a,b_x,c_x\geq 0  \,.
\]
One can check that the following choice satisfies these conditions for some $a^\star \geq 0$:
\begin{equation}
g_x = \tfrac12{a^\star p_x}  \mathbf{1}_{\{a^\star p_x/2<1\}}+\mathbf{1}_{\{a^\star p_x/2\geq1 \}},\quad
b_x =(a^\star p_x-2)\mathbf{1}_{\{a^\star p_x/2\geq1\}},\quad\,c_x=0\,.
\label{optsolqp}
\end{equation}
Hence $f_{\text{QP}}(P_X,\eps) =\sum_x \mathbf{1}_{\{a^\star p_x/2\geq1 \}}+\frac14{a^{\star 2}p_x^2}  \mathbf{1}_{\{a^\star p_x/2<1\}}$. 
Note further that the constraint is satisfied with equality, giving
\begin{align}
\label{eq:laterfeas}
\sqrt{1-\eps^2}=\sum_x \tfrac12{a^\star p_x^2}  \mathbf{1}_{\{a^\star p_x/2<1\}}+p_x\mathbf{1}_{\{a^\star p_x/2\geq1 \}}\,,
\end{align} 
which gives the alternate form 
\begin{align}
\label{eq:epsopt}
f_{\text{QP}}(P_X,\eps)=\tfrac12 a^\star \sqrt{1-\eps^2}-\sum_x (\tfrac12a^\star p_x-1)\mathbf{1}_{\{a^\star p_x/2\geq1 \}}\,.
\end{align}
Put differently, we have 
\begin{align}
\label{eq:simpleQP}
f_{\text{QP}}(P_X,\eps)=\min_a\left\{\tfrac12 a \sqrt{1-\eps^2}-\sum_x (\tfrac12ap_x-1)\mathbf{1}_{\{ap_x/2\geq1 \}}, a\geq 0\right\}\,.
\end{align}
This form is useful for numerically computing the optimal value, since now there is only a single variable.
Looking back at the definition, it is also clear that 
\begin{align}
\label{eq:QPhypo}
f_{\text{QP}}(P_X,\eps)\leq \beta_{\sqrt{1-\eps^2}}(P_X,\id_X)\,,
\end{align} 
since upper bounding the objective function in the QP by $\sum_x g_x$ yields the hypothesis testing optimization. This bound will be useful later. 

\begin{theorem}
	For all probability distributions $P_X$, 
	\begin{align}
	f_{\textup{SDP}}(P_X,\epsilon) = f_{\textup{QP}}(P_X,\eps) \,.
	\end{align}
	\end{theorem}

% \begin{equation} \label{qp}
% \lambda_{q}^{\epsilon}(\{p_k\}) \coloneqq \min_{0\leq g_k\leq 1}\left\{\sum g_k^2 \left\lvert \sum g_{k} p_k \geq \sqrt{1-\epsilon^2}\right.
% \right\}.
% \end{equation}
% \end{theorem}

\begin{proof}
The proof proceeds by establishing inequalities in both directions. 
For the upper bound, suppose that $g_x$ is optimal in the QP and define $\ket{g}=\sum_x g_x\ket{x}$, $\theta=\ketbra{g}$, and $\lambda=\braket{g|g}$. 
These are feasible in \eqref{spse}. 
Since $\theta$ is pure, the first condition is equivalent to $\lambda\geq \tr[\theta]$, which is clearly satisfied. 
The second is manifestly satisfied, while the third is just $\braket{\psi_X|g}\geq \sqrt{1-\eps^2}$, which is the first constraint in the QP. 
Hence $f_{\text{SDP}}(P_X,\epsilon) \leq  f_{\text{QP}}(P_X,\eps)$.

For the lower bound we construct a set of feasible variables for the dual of $f_{\text{SDP}}(P_X,\eps)$ from the optimal variables of $f_{\text{QP}}(P_X,\eps)$; this  requires a bit more effort. 
First, note that the dual of \eqref{spse} is given by 
\begin{equation} \label{dual}
	f_{\text{SDP-dual}}(P_X,{\epsilon}) \coloneqq \sup_{\mu,K,T}\{\mu (1-\epsilon^2)- \tr[K] \mid
	\mu  {\psi_X} -\mathsf{P}(K)- T \leq 0, \tr[ T]\leq 1,\mu\in \mathbb R_+,K,T\in\mathbb P_d
	\}\,.
\end{equation}
Since $f_{\text{SDP}}(P_X,{\epsilon})\geq f_{\text{SDP-dual}}(P_X,{\epsilon})$ by weak duality, to complete the proof it is sufficient to establish $f_{\text{SDP-dual}}(P_x,{\epsilon}) \geq  f_{\text{QP}}(P_X,\eps)$.  
Using $\ket{g}$ from before, set 
\[
\mu= \frac{a}{2\sqrt{1-\epsilon^2}}; T= \frac{\ketbra{g}}{\braket{g|g}}; K_{xx} = \mathbf{1}_{\{ap_x/2\geq 1\}}(\tfrac 12ap_x-1)\,,\] 
where $K$ is a diagonal matrix. 
Clearly the variables are all positive and $\tr[T]\leq 1$; using \eqref{eq:epsopt} implies that objective function satisfies 
$\mu(1-\eps^2)-\tr[K]=f_{\text{QP}}(P_X,\eps)$.

It only remains to show the first constraint holds. 
For this, consider the properties of $\mu\psi_X-\pinch(K)$. 
Since $\psi$ has rank one and $K$ is positive, the Weyl inequalities (see \cite[Theorem III.2.1]{bhatiaMatrixAnalysis1997}) imply the the second largest eigenvalue is negative. 
In detail, for $\lambda_j(M)$ the $j$th largest eigenvalue of $M$, by the Weyl inequalities we have $\lambda_2(\mu\psi_X-\pinch(K))\leq \lambda_2(\mu \psi_X)+\lambda_1(-\pinch(K))\leq 0$. 
On the other hand, $\ket{g}$ is an eigenvector of $\mu\psi_X-\pinch(K)$ with eigenvalue 1: Computing $\bra{x}(\mu\psi_X-\pinch(K))\ket{g}$ gives $\mu\sqrt{1-\eps^2}p_x-\mathbf{1}_{\{ap_x/2\geq 1\}}(\tfrac12 ap_x-1)=\tfrac12ap_x-\mathbf{1}_{\{ap_x/2\geq 1\}}(\tfrac12 ap_x-1)$, which is just $g_x$.  
Thus, $\mu  {\psi_X} -\pinch(K)$ has at most one positive eigenvalue. 
Moreover, by construction $T$ is equal to the positive part of $\mu \psi_X-\pinch(K)$, meaning $\mu \psi_X-\pinch(K)-T \leq 0$. 
Therefore, $f_{\text{SDP-dual}}(P_x,{\epsilon}) \geq  f_{\text{QP}}(P_X,\eps)$, which completes the proof. 
\end{proof}

For arbitrary pure states $\rho_{AR}$ with Schmidt coefficients $P_X$, we therefore have 
\begin{align}
H_{\min}^{\epsilon,P}(A\lvert\dot R)_{\rho_{AR}}=-\log f_{\text{QP}}(P_X,\eps)\,. 
\end{align}

\section{Asymptotic expansion}
\label{subsec:asym}
Now we turn to the asymptotic expansion of $H_{\min}^{\epsilon,P}(A\lvert\dot R)_{\rho_{AR}^{\otimes n}}$.
%Now that we have shown that for the pure state case, the SDP reduces to the much more tractable form of a quadratic program, we will use this simplified form to derive an asymptotic expansion for the partially smoothed conditional min-entropy. Moreover, we derive asymptotic expansions for other special cases, such as states with trivial $R$ as well as perfectly correlated classical states. 
% For comparison, recall that the globally smoothed conditional min-entropy has the following asymptotic expansion \cite[equation 6]{Tomamichel2013} \begin{equation}\label{gsexp}
% 	-H_{\min}^{\epsilon,P}(A\lvert {R})_{\rho_{AR}^{\otimes n}}= nS(\rho_A)+\sqrt{n V(\rho_A\lVert \mathbbm{1})}\Phi^{-1}\left({1-\epsilon^2}\right)+ \mathcal{O}(\log n).
% \end{equation}
%\subsection{Expansion for pure states}
%In the following theorem, we show an upper and a lower bound on the quadratic program (\ref{qp}) in terms of a quantity of the form $\beta_{\sqrt{1-\epsilon^2}}(\rho_A, \mathbbm 1)$. These bounds imply a second order asymptotic expansion for the pure state case. While proving the upper bound is fairly straightforward, the lower bound requires more work, as we will need to show that one of the terms in the quadratic program decays to zero fast enough so as not have an effect on the second order term, as explained in the proof of Theorem \ref{asymppure}.
\begin{theorem}\label{asymppure}
	For any pure $\rho_{AR}$, 
	\begin{align}
	\label{eq:partial2nd}
	-H_{\min}^{\epsilon,P}(A\lvert\dot R)_{\rho_{AR}^{\otimes n}} &= nH(A)_{\rho}+\sqrt{nV(A)_{\rho}}\Phi^{-1}(\sqrt{1-\epsilon^2}) +{O}(\log n)\,.
	\end{align}
	% \begin{IEEEeqnarray*}{rCl}
	% 	-H_{\min}^{\epsilon,P}(A\lvert\dot R)_{\rho_{AR}^{\otimes n}} &=& nS(\rho_A)+\sqrt{nV(\rho_A\lVert \mathbbm{1})}\Phi^{-1}(\sqrt{1-\epsilon^2}) +\mathcal{O}(\log n).
	% \end{IEEEeqnarray*}
\end{theorem}
\begin{proof}
Again the proof proceeds by establishing bounds in both directions. 
The upper bound is simple: Combining  \eqref{eq:QPhypo} and \eqref{eq:beta2nd} gives
\begin{align}
H_{\min}^{\epsilon,P}(A\lvert \dot{R})_{\rho_{AR}^{\otimes n}}\geq nD(\rho_A,\id_A)-\sqrt{n V(\rho_A,\id_A)}\Phi^{-1}(\sqrt{1-\eps^2})+O(\log n)\,.
%\log_2 \beta_{\sqrt{1-\epsilon^2}}(\rho_{A},\mathbbm{1})\,.
\end{align}
Since the conditional entropy is the negative of a relative entropy, this is the upper bound.

To establish the lower bound, it is sufficient to show that, for some $G$ which depends on $\rho_{AR}$ and not $n$, 
\begin{align}
\label{eq:target}
f_{\text{QP}}(P_X^{\otimes n},\eps) \geq \beta_{\sqrt{1-\epsilon^2}- \frac{G}{\sqrt{n}}}(P_X^{\otimes n},\mathbbm{1})\,.
\end{align}
For then, the second order expansion in \eqref{eq:beta2nd} and a Taylor expansion of $\Phi^{-1}$ (see, e.g.\ \cite[Lemma 3.7]{datta2015second}) gives 
\begin{align}
	 \log_2 f_{\text{QP}}(P_X^{\otimes n},\eps)
	 &\geq   nH(A)_\rho+\sqrt{nV(A)_\rho}\,\Phi^{-1}(\sqrt{1-\epsilon^2}-\frac{G}{\sqrt{n}}) +{O}(\log n) \\
	 &=nH(A)_\rho+\sqrt{nV(A)_\rho}\,\Phi^{-1}(\sqrt{1-\epsilon^2}) +{O}(\log n)\,.
\end{align}

To establish \eqref{eq:target}, suppose $a^\star$ is the optimizer in $f_{\text{QP}}(P_X,\eps)$ for an arbitrary $P_X$ and define $\eta^\star=\tfrac12 a^\star \sum_x p_x^2\mathbf 1_{\{a^\star p_x/2<1\}}$. 
Then by \eqref{eq:laterfeas}, $\Lambda$ with components $\Lambda_x=\mathbf1_{\{a^\star p_x/2\geq 1\}}$ is feasible for $\beta_{\sqrt{1-\eps^2}-\eta^\star}(P_X,\id_X)$. 
Since $f_{\text{QP}}(P_X,\eps)\geq \sum_x \mathbf1_{\{a^\star p_x/2\geq 1\}}$, we therefore have $f_{\text{QP}}(P_X,\eps)\geq \beta_{\sqrt{1-\eps^2}-\eta^\star}(P_X,\id_X)$. 
Replacing $P_X$ with $P_{X^n}=P_X^{\otimes n}$, it remains to show that $\eta^\star\leq \frac{G}{\sqrt n}$.
Let $Y^n$ be the random variable taking the value $P_{X^n}(x^n)$ when $X^n=x^n$. 
Then we can write $\eta^\star = \tfrac12 a^\star \mathbb E_{X^n}[Y^n\mathbf 1_{\{1/Y^n\geq a^\star/2\}}]$. 
For $Z=-\log_2 P_X(X)$, it follows that $\log Y^n=-\sum_{j=1}^n Z_j$, and therefore 
\begin{align}
\eta^\star = \tfrac12 a^\star \mathbb E_{X^n}[\exp(-\sum_{j=1}^n Z_j)\mathbf 1_{\{\sum_{j=1}^n Z_j \geq \log 2/a^\star\}}]\,.
\end{align}
Now we can appeal to Lemma 47 of \cite{polyanskiyChannelCodingRate2010}, which states that the expectation value is bounded above by $\frac{G}{\sqrt n}\exp(-\log c)$, where $G$ is a constant that depends on the sums of the second and third moments of $Z_j$ and $c$ is the lower bound of the indicator function inside the expectation, i.e.\ $\log \frac2{a^\star}$ in this case. Hence the sought-after bound $\eta^\star\leq \frac{G}{\sqrt{n}}$ holds, completing the proof.  
\end{proof}

\begin{remark}[Third order term]
	Using techniques similar to those of \cite[Appendix K]{polyanskiyChannelCodingRate2010}, one can derive the following bound on the third-order expansion for the qubit case.
	For $\rho_{AR}^{\otimes n} $ with $\rho_{A} =  \delta\ketbra 0+(1-\delta)\ketbra 1$\,, 
	\begin{align}
	-H_{\min}^{\epsilon,P}(A\lvert \dot{R})_{\rho_{AR}^{\otimes n}}\geq n h_2(\delta)+\sqrt{n v_2(\delta)}\Phi^{-1}\left(\sqrt{1-\epsilon^2}\right) - \tfrac{1}{2} \log_2 n+ {O}(1)\,,\label{eq:3rdorder}
	\end{align}
	where $h_2$ is the binary entropy $h_2(\delta)=-\delta\log \delta-(1-\delta)\log(1-\delta)$ and $v_2$ the binary variance $v_2(\delta)=\delta(1-\delta)(\log \frac{1-\delta}\delta)^2$. 

\end{remark}

%\subsubsection{Other Special Cases}

In light of the above, one may expect that the asymptotic expansion to have the same form for arbitrary states. 
However, this is not the case, as demonstrated by the following two examples. 
\begin{prop}
If $\rho_{AR}=\sum_x p_x\ketbra{x}_A \otimes \ketbra{0}_R$ or $\rho_{AR}=\sum_x p_x \ketbra{x}_A\otimes \ketbra{x}_R$ for some probability distribution $p_x$, then $H_{\min}^{\epsilon,\Delta}(A\lvert \dot{R})_{\rho_{AR}}=H_{\min}^{\epsilon,\Delta}(A\lvert {R})_{\rho_{AR}}$ for any distance measure $\Delta$.
\end{prop}
\begin{proof}
%First note that if if $p_x$ is deterministic, i.e.\ for a particular value of $x$, $p_x=1$ and all others are zero, then both the entropy and variance are zero, so the partial and global asymptotic expansions match.  
In the former case $R$ plays no role, i.e.\ we can set $d_R=1$, and therefore partial and global smoothing are equivalent. 
For the latter, first note that the optimal $\sigma_{AR}$ has the form $\sigma_{AR} = \sum_x s_x \ketbra{xx}$. 
This follows because applying the map $\sum_x \ketbra{xx}(.)\ketbra{xx}$ to any feasible $\sigma_{AR}$ results in another feasible solution of the desired form and with the same $\lambda$. 
Then, observe that the constraint $\lambda \mathbbm{1}_A \otimes \rho_R\geq \sigma_{AR}$ in this case reduces to $\lambda p_x \geq s_x$ for all $x$. 
Setting $s_x=p_x$ is certainly feasible, and implies that the optimal $\lambda$ is smaller than 1. 
However, this means the first constraint implies the partial smoothing constraint, and therefore  any optimal solution for the globally smoothed quantity is feasible (and optimal) for the partially smoothed quantity.
\end{proof}

Note that these examples are no contradiction to the above pure state results. 
The intersection of the two sets are pure product states; their second order expansion is zero, as both the entropy and variance vanish.

Another general case of equivalence is that of arbitrary classical states with the trace distance metric. 
The proof is entirely similar to the collision entropy case in \cite{Yang2017}.
\begin{prop}
	For an arbitrary probability distribution $P_{XY}$, $H_{\min}^{\epsilon,T}(X\lvert \dot{Y})_{P}=H_{\min}^{\epsilon,T}(X\lvert {Y})_{P}$.
\end{prop}
\begin{proof}
	For this proof, we use the trace distance in the SDP form discussed in Remark \ref{SDPform}, so that the constraint $T(P_{XY},S_{XY}) \leq \epsilon$ is replaced by the two constraints $W_{XY}+S_{XY}\geq P_{XY}, \tr[W_{XY}]\leq \epsilon$.
	Note that the positivity constraint reduces to pointwise positivity. 
	Now, define the following SDP as a relaxation of the globally smoothed conditional min-entropy SDP (\ref{gse}) with trace distance
	\begin{align}
	\lambda_{r}^{\epsilon,T}(X\lvert {Y})_{P} \coloneqq \min_{\lambda,W_{XY}\geq 0}\{\lambda \mid \lambda \mathbbm{1}_X \otimes P_Y\geq P_{XY}-W_{XY},\tr[W_{XY}]\leq \epsilon\}\,.
	\end{align}
	Since these SDPs form a sequence of relaxations, we have $\lambda_{\max}^{\epsilon,T}(X\lvert \dot{Y})_{P}\geq \lambda_{\max}^{\epsilon,T}(X\lvert {Y})_{P}\geq \lambda_r^{\epsilon,T}(X\lvert {Y})_{P}$. 
	It remains to show that from an optimal solution $(\lambda_r^\star,W_{XY}^\star)$ of $\lambda_r^{\eps,T}$, one can construct a feasible solution for the partially smoothed SDP with the same objective value. 
	Observe that the optimal $W_{XY}^\star$ is smaller than $P_{XY}$, since for any feasible $W_{XY}$, $W'_{XY}$ with $W'_{XY}(x,y)=\min\{P_{XY}(x,y),W_{XY}(x,y)\}$ is also feasible. 
	Hence $S_{XY}=P_{XY}-W^\star_{XY}$ is feasible for the partially smoothed SDP, and indeed the more stringent version of Renner and Wolf, which requires $S_{XY}\leq P_{XY}$.  
	% This is possible by choosing $W^\prime_{XY}= \min\{W_{XY}^\star,\rho_{XY}\}$. %Note that if $(\lambda_r^*,W_{XY}^*)$ is feasible for the relaxed SDP, then $(\lambda_r^*,W^\prime_{XY})$ is feasible too (observe that $W_{XY},\rho_{XY}$ are diagonal, and $\lambda,\rho_Y\geq0$).
	%  Then, we set $\sigma_{XY}=\rho_{XY}-W^\prime_{XY}$, which is feasible for the partially smoothed SDP, since $\sigma_{XY}\geq 0$ and  $\sigma_{Y}\leq \rho_Y$. 
	%  Thus,  $\lambda_{\max}^{\epsilon,T}(X\lvert \dot{Y})_{\rho_{XY}}= \lambda_{g}^{\epsilon,T}(X\lvert {Y})_{\rho_{XY}}= \lambda_r^{\epsilon,T}(X\lvert {Y})_{\rho_{XY}}$.
	\end{proof}
	The fact that the positivity constraint reduces to pointwise comparison is crucial for the argument. 
	In the quantum case one cannot always find an optimal $W_{AR}^\star$ such that $\rho_{AR}-W_{AR}^\star\geq 0$. 
	A counterexample is given by the state $\ket{\psi}=\tfrac1{\sqrt{10}}(3\ket{00}+\ket{11})$ at $\eps=0.1$.
	The equivalence for classical states also relies on using the trace distance. 
	A counterexample for purification distance is given by $P_{XY}$ with $P_{XY}(0,y)=1/4$, $P_{XY}(1,0)=0$, and $P_{XY}(1,1)=1/2$ at $\eps=0.1$. 
	
	We also show equivalence of the partial and global smoothing of the max mutual information (see \cite{Anshu2018}) in the Appendix. Therefore, the second order asymptotics of classical partially-smoothed conditional min-entropy and max mutual information are determined by their globally-smoothed counterparts.

 %!TEX root = paper.tex
\section{Application to quantum data compression}\label{sec:qc}

%% TO DO :  How is this different from visible? from classical? ...Different compression setups.
In this section we determine the optimal second order rate of quantum data compression. 
First studied by Schumacher~\cite{Schumacher1995}, the task of quantum data compression is to map a fixed quantum state $\rho_A$ to a Hilbert space of smaller dimension, such that the original state can be approximately recovered. 
Traditionally, the approximation quality is measured by the entanglement fidelity: For approximation parameter $\eps$ and some purification $\rho_{AR}$ of $\rho_A$, the recovered state $\rho'_{AR}$ should satisfy $F(\rho_{AR},\rho'_{AR})^2\geq 1-\eps$ (note the square). 
As defined in \cite[Definition 5.4]{datta2015second}, an $(M,\eps)$ compression protocol consists of a channel $\cC$ from $A$ to a system of $M$ qubits (i.e.\ of dimension $2^M)$ and a decompression channel back to $A$, such that $F(\rho_{AR},\cD\circ \cC(\rho_{AR}))^2\geq 1-\eps$. 
We denote by $M^\star(\rho_A,\eps)$ the smallest possible $M$ for a given combination of $\rho_A$ and $\eps$. 
Then we can show
\begin{theorem}
\label{thm:qc2nd}
For any $\eps\in (0,1)$ and any state $\rho_A\in \mathbb P_d$, 
\begin{align}
M^\star(\rho_A^{\otimes n},\eps) = n H(A)_\rho+\sqrt{n V(A)_\rho}\Phi^{-1}(\sqrt{1-\eps})+O(\log n)\,.
\end{align}
\end{theorem}
The proof proceeds by finding lower and upper bounds whose asymptotic expansions match at second order. 
The upper bound (achievability) is a very slight improvement over \cite[Theorem 5.5(ii)]{datta2015second}, while for the lower (converse) bound we make use of the connection between compression and state merging, and then use the converse bound on state merging in terms of partially smoothed conditional min-entropy from \cite[Theorem 6]{Anshu2018}. 
Before delving into the details of the proof, let us first examine the relationship between compression and state merging. 

As introduced in~\cite{horodeckiPartialQuantumInformation2005,horodecki_quantum_2007}, in quantum state merging two parties Alice $(A)$ and Bob $(B)$ share a state $\rho_{AB}$ with purification $\rho_{ABR}$, and the goal is to transfer the $A$ system to Bob by using only classical communication and shared entanglement. 
In general, one is interested in minimizing both the communication and entanglement costs of the protocol, but here we need only consider the entanglement cost. 
Following \cite{Anshu2018}, an $(E,\eps)$ protocol for state merging consists of a quantum channel $\cE$ from $A$ and $A_0$ to $X$ and $A_1$ with $X$ classical, a quantum channel $\cF$ from $X$, $B$, and $B_0$ to $B$, $\bar B$, and $B_1$ such that 
\begin{align}
P(\rho_{\bar BBR}\otimes \Phi_{A_1B_1},\cF\circ \cE(\rho_{ABR}\otimes \Phi_{A_0B_0}))\leq \eps\,,
\end{align}
where $\rho_{\bar BBR}$ is just $\rho_{ABR}$ with $\bar B$ replacing $A$, $\Phi$ is the maximally entangled state, and $E=\log |A_0|-\log |A_1|$. 
%Note that the use of the purification distance means the entanglement fidelity of the output should be larger than $1-\eps^2$. 
The optimal entanglement cost is denoted $E^\star(\rho_{AB},\eps)$.

For states with trivial $B$, any compression protocol can be used to achieve state merging, simply by combining it with teleportation.  
In particular, Alice can first compress $\rho_A$ to $M$ qubits and then transfer these to Bob using $M$ entangled pairs and $2M$ bits of classical communication. 
Then he can implement the decompressor to recover the state. 
Owing to the slightly different approximation parameters used in the two definitions, we have
\begin{prop}
\label{prop:qsmqc}
For any state $\rho_A$, if there exists an $(M,\eps)$ compression protocol, then there exists an $(M,\sqrt{\eps})$ state merging protocol.
\end{prop}

%Now we proceed to the proof.
\begin{proof}[Proof of Theorem~\ref{thm:qc2nd}]
A very slight modification of the proof of \cite[Theorem 5.5(ii)]{datta2015second} gives 
\begin{align}
\cl{\rho_A}{\epsilon}\leq \bar{H}_s^{1-\sqrt{1-\epsilon}}(A)_\rho\,.
\label{improvedachievability}
\end{align}
Set $\gamma = \bar{H}_s^{1-\sqrt{1-\epsilon}}(A)_\rho$, and $Q = \{\rho_A \geq 2^{-\gamma}\mathbbm 1\}.$ 
Pick the compression map in to be $\mathcal{C}(\rho_A) \rightarrow Q\rho_A Q +\tr[\rho_A(\mathbbm 1 -Q)]\ketbra{\phi} $, for an arbitrary state $\ketbra{\phi}$ in the image of the projector $Q$. 
The compressed system is the image of $Q$ with dimension \begin{equation}
M = \tr[Q] \leq 2^\gamma \label{eq:M}.
\end{equation}
The decompressor simply embeds this system back into $A$. 
Then one can show that the entanglement fidelity is at least $1-\epsilon$ as follows:
	\[F(\rho_{AR}, \mathcal{D}\circ \mathcal{C}(\rho_{AR}))^2 \geq \left(\tr[Q\rho_A]\right)^2 \geq \left( \tr[Q(\rho_A-2^{-\gamma}\mathbbm 1)]\right)^2
	\geq (1-(1-\sqrt{1-\epsilon}))^2. %% todo : Should I add in more of their proof details??	
	\]The last step follows from the definition of $\bar{H}_s^{\epsilon}(A)_\rho$ and the choice of $\gamma$.

Combining \eqref{improvedachievability} and \eqref{eq:spec2nd} gives the bound 
\begin{align}
M^\star(\rho_A^{\otimes n},\eps) \leq  n H(A)_\rho+\sqrt{n V(A)_\rho}\Phi^{-1}(\sqrt{1-\eps})+O(\log n)\,,
\end{align}
where we have used $\Phi^{-1}(\eps)=-\Phi^{-1}(1-\eps)$.

Meanwhile, Proposition~\ref{prop:qsmqc} implies $E^\star(\rho_A,\sqrt{\eps})\leq M^\star(\rho_A,\eps)$, and so by \cite[Theorem 6]{Anshu2018} we have 
\begin{align}
-H_{\min}^{\sqrt\epsilon,P}(A\lvert\dot R)_\rho\leq M^\star(\rho_A,\eps)\,.\label{eq:qcconverse}
\end{align}
Combining this with \eqref{eq:partial2nd} gives
\begin{align}
M^\star(\rho_A^{\otimes n},\eps) \geq  n H(A)_\rho+\sqrt{n V(A)_\rho}\Phi^{-1}(\sqrt{1-\eps})+O(\log n)\,,
\end{align}
and the proof is complete. 
\end{proof}

Given the gap in the bounds of \cite{datta2015second}, one may have thought that a more sophisticated compression protocol would be needed to match the converse therein; something involving more than just a classical protocol applied to the eigenbasis. However, the improved state merging converse shows that this is not the case, and the straightforward compression protocol is sufficient to achieve the optimal rate at second order.

Figure \ref{fig:rateplt} shows a numerical evaluation of the compression bounds obtained in this paper for the qubit case. The original compression converse, the achievability bound in terms of the information spectrum entropy and a direct achievability bound through equation (\ref{eq:M}) are also shown on the plot for comparison. 
Using the form in \eqref{eq:simpleQP}, one can evaluate the partially-smoothed conditional min-entropy by means of a simple binary search. 
%Since the partially smoothed conditional min-entropy SDP is reduced to a quadratic program, it becomes possible to evaluate it, without the use of any optimization packages, by means of a simple binary search program. 
\begin{figure}[]
	\centering
	\scalebox{0.7}{
		
	  \centering
  \pgfplotscreateplotcyclelist{my black white}{%
  	solid, every mark/.append style={solid, fill=gray}, mark=*\\%
  	dotted, every mark/.append style={solid, fill=gray}, mark=square*\\%
  	densely dotted, every mark/.append style={solid, fill=gray}, mark=otimes*\\%
  	loosely dotted, every mark/.append style={solid, fill=gray}, mark=triangle*\\%
  	dashed, every mark/.append style={solid, fill=gray},mark=diamond*\\%
  	loosely dashed, every mark/.append style={solid, fill=gray},mark=*\\%
  	densely dashed, every mark/.append style={solid, fill=gray},mark=square*\\%
  	dashdotted, every mark/.append style={solid, fill=gray},mark=otimes*\\%
  	dashdotdotted, every mark/.append style={solid},mark=star\\%
  	densely dashdotted,every mark/.append style={solid, fill=gray},mark=diamond*\\%
  }
  \pgfplotstableread{new_ratepltpaper1000step50.dat}\plotdata

\begin{tikzpicture}

\begin{axis}[height=\textwidth,width=\textwidth,
xmin=50,xmax=1000,ymin=0.5,ymax=0.64,xlabel=$n$,
title={$\epsilon = 0.1, \delta = 0.9$},legend style={draw=none},legend pos=north east,cycle list name=my black white]
\addplot table[x index=0,y index=4] {\plotdata};
\addplot table[x index=0,y index=5] {\plotdata};
\addplot table[x index=0,y index=1] {\plotdata};

\addplot table[x index=0,y index=3] {\plotdata};

\legend{
		$\frac{1}{n}\bar H_s^{1-\sqrt{1-\epsilon}}(\rho_A)$ { info. spectrum achiev. (proof of Thm \ref{thm:qc2nd})},
			$\frac{\log_2 M}{n}$ { direct achievability ($M = \tr[Q] $ equation (\ref{eq:M}))},
		$\frac{-1}{n}H_{\min}^{\sqrt\epsilon,P}(A\lvert \dot{R})${ new converse, equation (\ref{eq:qcconverse})},
		$\frac{1}{n}(\bar H_s^{\epsilon+\eta}(\rho_A)+\log_2 \eta)${ info. spectrum converse \cite[Theorem 5.5]{datta2015second}}
},

\end{axis}

\end{tikzpicture}}
	\caption{\label{fig:rateplt}Compression bounds for the qubit case $\rho_{A} =  \delta\ketbra 0+(1-\delta)\ketbra 1$ }  
\end{figure}

\vspace{2mm}
{\bfseries{Acknowledgments.}} 
DA acknowledges the generous support of the INSPIRE Potentials - QSIT
Master Internship Award. JMR acknowledges support by the Swiss National Science Foundation (SNSF) via the National Centre of Competence in Research “QSIT”, as well as the Air Force Office of Scientific Research (AFOSR) via grant FA9550-16-1-0245.
\printbibliography[heading=bibintoc,title=References]

\appendix
%!TEX root = paper.tex
\section{Appendix}

Here we show that the partially and globally smooth max mutual informations are equivalent for classical distributions under trace distance smoothing. 
Following \cite{Anshu2018}, let us first define two optimizations 
\begin{align}
\lambda_{\max}^\epsilon(X{:}Y)_P
&\coloneqq \underset{Q,T,R\geq 0}{\text{inf}}\{\tr[R_Y]
\mid P_XR_Y \geq Q_{XY},
T_{XY}+Q_{XY}\geq P_{XY},
\tr[T_{XY}]\leq \epsilon,
Q_X=P_X\}\quad\text{and}\\
\bar\lambda_{\max}^\epsilon(X{:}Y)_P&\coloneqq \underset{Q,T,R\geq 0}{\text{inf}}\{\tr[R_Y]
\mid P_XR_Y \geq Q_{XY},
T_{XY}+Q_{XY}\geq P_{XY},
\tr[T_{XY}]\leq \epsilon,
\tr[Q_{XY}]=1\}\,.
\end{align}
The the partially smoothed quantity is defined as ${I_{\max}^\eps(\dot{X}{:}{Y})_P}=\log\lambda_{\max}^\epsilon(X{:}Y)_P$, while the globally smooted quantity is 
${I_{\max}^\eps(X{:}{Y})_P}=\log \bar\lambda_{\max}^\epsilon(X{:}Y)_P$.

The latter is a relaxation of the former, so $\lambda_{\max}^\epsilon(X{:}Y)_P\geq \bar\lambda_{\max}^\epsilon(X{:}Y)_P$. 
Just as before, our strategy is to construct a feasible set of variables for the former from the optimal variables of the latter. Call these $(Q_{XY}^\star,T_{XY}^\star,R_{Y}^\star)$.
First note that $\tr[R_Y^\star]\geq 1$, which follows from normalization of $P_{XY}$ and the first and fourth constraints. 
Repeating the argument for the min entropy, we again have $T_{XY}^\star\leq P_{XY}$.
Combining the first and second constraints to remove $Q_{XY}^\star$ gives $P_XR_Y^\star\geq P_{XY}-T_{XY}^\star$. 
This is satisfied for all $x$ such that $P_X(x)=0$, since for all such $x$, $T_{XY}^\star(x,y)=0$ for all $y$. 
Thus if we write $T_{XY}^\star=P_XT^\star_{Y|X}$ for some $T_{Y|X}^\star$, the combined constraint simplifies to $R_Y^\star\geq P_{Y|X=x}-T^\star_{Y|X=x}$ for all $x$ such that $P_X(x)>0$. 
The quantity on the left is supernormalized, the quantity on the right subnormalized, so we may pick a normalized $\tilde Q_{Y|X=x}$ which lies somewhere in between. 
Setting $\tilde Q_{XY}=P_X\tilde Q_{Y|X}$, we have $\tilde Q_X=P_X$ by construction, and therefore $(\tilde Q_{XY},T_{XY}^\star,R_Y^\star)$ is feasible in $\lambda_{\max}^\epsilon(X{:}Y)_P$.

On the other hand, the counterexample from the conditional min-entropy also shows the  nonequivalence of partially and globally smoothed max mutual information in the quantum case. 

% \subsection{Connection between Information Spectrum Entropy and $E_\gamma$}

%Moreover, as in \cite{JMRPA2017},
%\[E_\gamma(\rho,\sigma) = \alpha (\gamma)-\gamma \beta_{\alpha (\gamma)}(\rho,\sigma).\]
% 	%%% TO DO : Proof? is this true ???

\end{document}